\documentclass[%
 aip,
 amsmath,amssymb,
 reprint,%
]{revtex4-1}

\usepackage{graphicx}% Include figure files
\usepackage{dcolumn}% Align table columns on decimal point
\usepackage{bm}% bold math
\usepackage[utf8]{inputenc}
\usepackage[T1]{fontenc}
\usepackage{mathptmx}
\usepackage{etoolbox}

\usepackage{amssymb}
\usepackage{array}
\usepackage{newtxtext}
\usepackage{palatino}
\usepackage{indentfirst}
\usepackage{booktabs}
\usepackage{float}

\usepackage{epstopdf}

\usepackage{graphicx}% Include figure files
\usepackage{dcolumn}% Align table columns on decimal point
\usepackage{bm}% bold math
\usepackage{comment}

\usepackage{amsthm}
\usepackage{amsmath}
\makeatletter

\newtheorem{theorem}{Theorem}
\newtheorem{lemma}{Lemma}

\newcommand*\bigcdot{\mathpalette\bigcdot@{.5}}
\newcommand*\bigcdot@[2]{\mathbin{\vcenter{\hbox{\scalebox{#2}{$\m@th#1\bullet$}}}}}

\def\@email#1#2{%
 \endgroup
 \patchcmd{\titleblock@produce}
  {\frontmatter@RRAPformat}
  {\frontmatter@RRAPformat{\produce@RRAP{*#1\href{mailto:#2}{#2}}}\frontmatter@RRAPformat}
  {}{}
}%
\makeatother
\begin{document}

\preprint{AIP/123-QED}

% Main title of the paper
\title{The impact of nodes of information dissemination on epidemic spreading in dynamic multiplex networks}    
\author{Minyu Feng}

\author{Xiangxi Li}

%\affiliation{College of Artificial Intelligence, Southwest University, Chongqing 400715, P. R. China.}
\author{Yuhan Li}
\affiliation{ 
College of Artificial Intelligence, Southwest University, Chongqing 400715, P. R. China.}
%\affiliation{College of Artificial Intelligence, Southwest University, Chongqing 400715, P. R. China.}
\author{Qin Li}
\email{qinli1022@gmail.com}
\affiliation{School of Public Policy and Administration, Chongqing University, Chongqing 400044, P. R. China.}

\date{\today}

% Here goes the abstract
\begin{abstract}
Epidemic spreading processes on dynamic multiplex networks provide a more accurate description of natural spreading processes than those on single layered networks. To describe the influence of different individuals in the awareness layer on epidemic spreading, we propose a two-layer network-based epidemic spreading model including some individuals who neglect the epidemic, and we explore how individuals with different properties in the awareness layer will affect the spread of epidemics. The two-layer network model is divided into an information transmission layer and a disease spreading layer. Each node in the layer represents an individual with different connections in different layers. Individuals with awareness will be infected with a lower probability compared to unaware individuals, which corresponds to the various epidemic prevention measures in real life. We adopt the micro Markov chain approach (MMCA) to analytically derive the threshold for the proposed epidemic model, which demonstrates that the awareness layer affects the threshold of disease spreading. We then explore how individuals with different properties would affect the disease spreading process through extensive Monte Carlo numerical simulations. We find that individuals with high centrality in the awareness layer would significantly inhibit the transmission of infectious diseases. Additionally, we propose conjectures and explanations for the approximately linear effect of individuals with low centrality in the awareness layer on the number of infected individuals.
\end{abstract}

\maketitle

\begin{quotation}
Human societies have suffered from various epidemics all the time. To effectively predict and control the spread of epidemics, lots of researchers use empirical analysis or mathematical modeling to study the spread mechanisms of diseases. The study of epidemic spreading on networks can provide a new perspective to control the spread of epidemics in society. Our task is to discuss how nodes with various topology properties in the awareness layer will affect the spread of diseases. Therefore, we establish a two-layer network propagation model with nodes that do not react with other nodes in the awareness layer. In addition, we analyze the propagation thresholds under this model and provide some explanations. We find significant differences in the impact of nodes with different centrality on disease spreading, nodes with a low centrality in the awareness layer have less effect on the spread of disease. Our work sheds some new lights on the propagation on awareness-disease multiplex networks.
\end{quotation}

\section{Introduction}\label{sec:introduction}

The spread of epidemics is constantly endangering human health. Research on theoretical modeling of epidemic spreading began in 1760, when Swiss mathematician Bernoulli studied the effectiveness of smallpox inoculation through mathematical methods. Kermack and Mckendrick established the threshold theory of epidemic spreading for the first time\cite{1}. In recent decades, with the rapid development of data science, the study of epidemic spreading has ushered in a golden period. Unlike the traditional research methods of epidemic spreading, the current stage of research work mainly considers the structure of interactions between individuals\cite{2}, contact patterns\cite{3}, and migration movement patterns of groups\cite{4}. It has been found that many natural social systems do not satisfy the assumption of homogeneous mixing, both at the individual and population levels, and exhibit the properties of a complex network. Nowadays, modeling research based on complex networks has become a hot spot in epidemic spreading research. Pastor et al. helped us understand computer virus epidemics and other spreading phenomena on communication and social networks \cite{5}. Saumell et al. developed a heterogeneous mean-field approach that allows us to calculate the conditions for the emergence of an endemic state \cite{6}. Stegehui et al. found that community structure is an important determinant of the behavior of percolation processes on networks\cite{7}. The theoretical framework of complex networks cannot only grasp some essential characteristics of natural complex systems but also perform rigorous mathematical calculations. By using the ideas and methods of statistical physics and other disciplines, researchers have developed many different theoretical approaches\cite{8} to study the transmission on complex networks, and many results have been achieved.

In recent years, there has been an increasing number of articles utilizing physical methods to study various social phenomena\cite{9}, accompanied by a research boom in the exploration of the structural aspects of multiplex networks. Li et al. proposed a mathematical framework for the coevolution of epidemic and infodemic on higher-order networks described by simplicial complex, and introduce the Microscopic Markov Chain Approach\cite{10}. This method will also be used in this article. Nie et al. improved Microscopic Markov Chain Approach to solve simplicial competing spreading involving pairwise and high-order interactions\cite{11}. Spreading processes in multiplex networks is still a nascent research area that presents numerous challenging research issues \cite{12}. Li et al. proposed two models from individual and population perspectives and applied stochastic methods to analyze both models \cite{13}. In \cite{14}, Li et al. proposed a limit for degree growth in a new network model. The study of epidemic spreading on multiplex coupled networks has received increasing attention \cite{15, 16}. Zhan et al. proposed an information-driven adaptive model in epidemic spreading, they showed that information and adaptive process can inhibit epidemic spreading\cite{17}. Hong et al. similarly using the micro Markov chain showed that time-varying behavioral responses can effectively suppress the epidemic spread with an increased epidemic threshold\cite{18}. The control of propagation on multiplex networks is also a popular research area. Jiang et al. found that the threshold is dependent upon both the connection strength between the layers and their internal structure \cite{19}. Sun et al. proposed a multilayer network model to study the impact of resource diffusion on disease propagation in such higher-order networks. They showed that the final fractions of infected individuals obtained via the micro Markov chain(MMC) method and the classical Monte Carlo method are very similar \cite{20}. Guo et al. found that when two nodes correspond, the one who knows about prevention takes effective measures to avoid infection. Comparison of MMC and Monte Carlo simulation results showed high consistency, indicating that MMC can predict epidemic outbreaks\cite{21}. Epidemic spreading models based on multiplex coupled networks can more accurately portray the process of disease spreading in the real world. When an epidemic breaks out, information about people's awareness and discussions about the disease will also spread on the Internet. The spread of an epidemic will intensify the spread of information, and the spread of information will remind people to take preventive measures, such as reducing contact with others or vaccination, which can inhibit the further spread of the disease \cite{22, 23}. In such networks, Fan et al. introduced 2-simplex interactions in the information layer, and they proved that this approach could be used to inhibit epidemic outbreaks \cite{24}. Li et al. proposed a modified signed-susceptible-infectious-susceptible epidemiological model that incorporates positive and negative transmission rates based on structural balance theory \cite{25}.

However, existing studies have focused on the impact of nodes on propagation based on their additional attributes, and it is assumed that after receiving epidemic information, nodes will reduce their probability of self-infection in some way. But recent studies and real-life scenarios have shown that some members of society do not take defensive measures against diseases \cite{26}. In real life, these individuals may not believe in the existence of the epidemic, be negligent in prevention after knowing the existence of the epidemic, or unable to change their behavior patterns. Clearly, such individuals cannot be included in the A-state. Therefore, this paper explores the impact of these individuals on epidemic propagation and investigates the effects of nodes with different characteristics as such individuals through a rich set of experiments, ultimately arriving at conclusions. To address this issue, we establish a two-layer network model with nodes that do not take any defensive measures and explore the impact of such nodes on disease spreading by comparing the selection patterns of different nodes, which we call $\Omega$ nodes. Using our model, we analyze the threshold by utilizing the micro Markov chain approach (MMCA). We also conduct extensive numerical simulations, the results of which show that different properties of individuals in the awareness layer can have varying effects on disease propagation. We discuss each specific property in detail and find that the state of individuals with high centrality in the awareness layer is crucial for disease spreading, while individuals with low centrality have less impact on disease propagation.

The paper is organized as follows: In Sec. \uppercase\expandafter{\romannumeral2}, we introduce the process of model building. In Sec. \uppercase\expandafter{\romannumeral3}, we derive the theoretical threshold of the model by MMCA. In Sec. \uppercase\expandafter{\romannumeral4}, we perform numerical simulations. Finally, discussions, conclusions, and outlooks are given in Sec. \uppercase\expandafter{\romannumeral5}.

\section{Model Description}

Most of the standard two-layer network propagation models overlook the problem, and they assume that every node aware of the disease will necessarily take appropriate preventive measures against the disease, which is not very reasonable in real life since there are people who do not take any preventive measures even though they are informed about the infectious disease. For example, when COVID-19 spread worldwide, some people refuse to wear masks\cite{26}, and some people think the epidemic is fake news\cite{27}. Or they relax their vigilance after the disease has spread for a period. In short, there are always some people who are aware of the existence of infectious diseases but still do not change their behavior patterns, and when they meet the conditions of awareness spreading, we cannot classify them as Awareness state (A-state) nodes, since although A-state nodes are called awareness nodes, their essence (impact on the transmission process) is to reduce the risk of being infected by the virus.

\begin{figure}[ht]
  \centering
  \includegraphics[scale = 0.38]{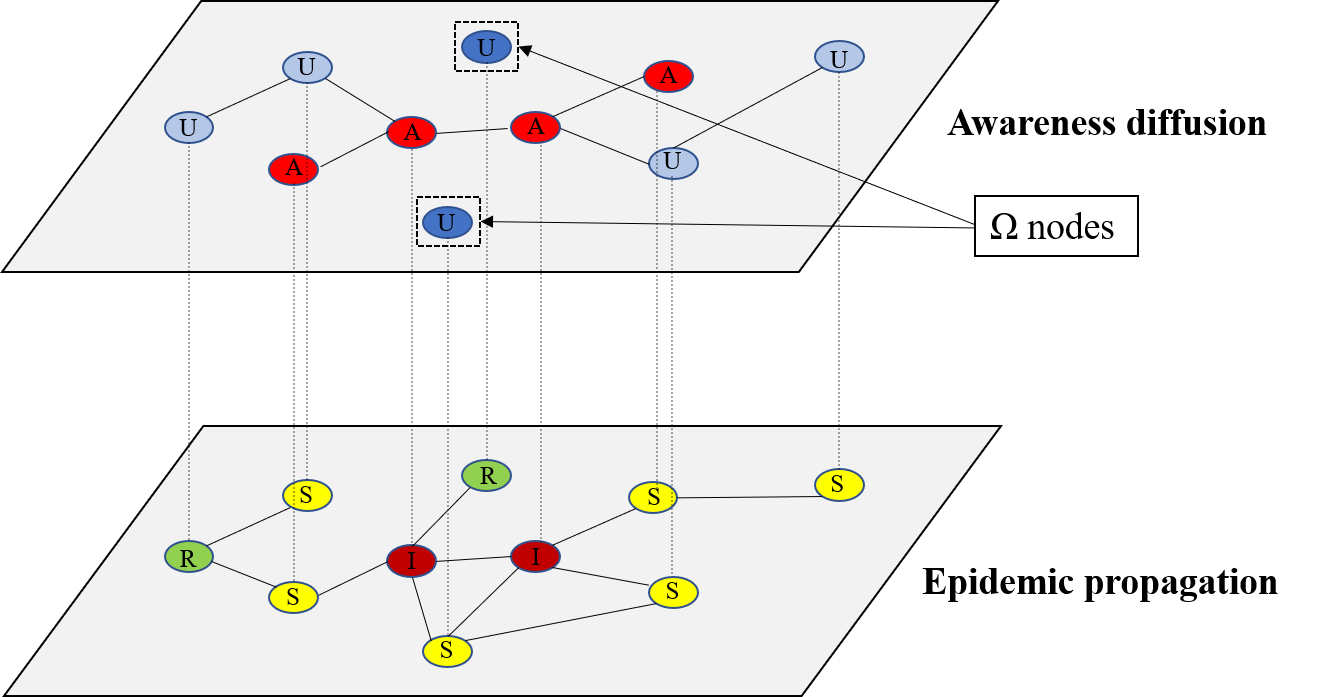}
  
  \caption{\textbf{The structure of a two-layer network.} The states of the nodes have been marked inside the nodes, the dark blue nodes represent $\Omega$ nodes, they do not interact with the nodes in the awareness layer and are always in the U-state. The figure demonstrates that two layer networks have different topological structures, but nodes are one-to-one corresponding.}
\end{figure}

In contrast to previous models that assumed every individual who received information would take measures to reduce their infection rate, our model takes into account the existence of individuals who are aware of the disease but fail to take appropriate preventive measures. These individuals, which we refer to as $\Omega$ nodes in our model, are a novel addition inspired by the aforementioned facts. Then, we analyze the transmission threshold of this model by the MMCA, and select nodes with a large degree, small degree, large betweenness centrality, and small betweenness centrality respectively as $\Omega$ nodes by the numerical simulation, and compare the simulation results with randomly selected $\Omega$ nodes. And we focus on the effect of different $\Omega$ nodes on the propagation process.

\begin{table}[ht]
    \centering
    \caption{Description of symbols}
    \begin{tabular}{ll}
    \toprule
    Symbols & Description\\
    \midrule
    $A$& The adjacency matrix of the \\
    $ $& information diffusion network\\
    $B$& The adjacency matrix of the \\
    $ $& disease propagation network\\
    $a_{ij}$& Elements of the matrix $A$\\
    $b_{ij}$& Elements of the matrix $B$\\
    $r_i(t)$&The probability of node $i$ is not informed at time $t$\\
    $q_i^A(t)$& The probability that A-state node $i$ is \\
    $ $ & not infected at time step $t$\\
    $q_i^U(t)$& The probability that U-state node $i$ is \\
    $ $& not infected at time step $t$\\
    $r$& Inability to disseminate information rate\\
    $\lambda$ & Information dissemination rate\\
    $\delta$& Information recovery rate\\
    $\beta^A$& A-state node disease spreading rate\\
    $\beta^U$& U-state node disease spreading rate\\
    $\mu$& Disease recovery rate\\
    $\gamma$& Rate of disease suppression \\
    $ $ & by information dissemination\\
    $\Omega$& Nodes that do not cooperate with epidemic prevention\\
    
    \bottomrule%第三道横线
    \end{tabular}
  \end{table}

 Hereby, we display the construction process of our model. The basic structure is a two-layer network. The information propagation network is modeled as a BA scale-free network primarily due to the fact that in information networks, the majority of nodes have degrees lower than the average degree, while a small number of nodes have degrees significantly higher than the average degree, such as news media outlets, which is a typical characteristic of BA networks. On the other hand, the disease transmission network is a physical contact network, and the small-world (WS) network is more suitable for it. In accordance with the work already done by previous authors, the coupled spread of disease and information is usually studied in a two-layer network model. In our model, epidemics occur in the physical contact layer, while information spreads through the virtual layer. The nodes of these two networks are one-to-one, but their respective connections are different, that is, the topology of the two layers is different. Let the upper network be the virtual contact network, where the information spreads. Let the lower network be the physical contact network, where the epidemic spread occurs. As shown in Fig. 1. There is a class of $\Omega$ nodes in the information layer, which only participate in the virus propagation process in the physical contact layer, not participating in the virtual layer information propagation process. In the virtual layer, $\Omega$ nodes are removed, while in the physical contact layer, their corresponding information state is equivalent to the unawareness state (U-state).

The diffusion of information in the virtual awareness layer follows the unaware-aware-unaware (UAU) model, where ``U'' denotes the unaware state, i.e., this node is unaware of the risk of infection or artificially chooses to ignore the risk of infection and does not take any preventive measures. ``A'' denotes the aware state, where nodes in this state are aware of the risk of disease and thus take certain measures to reduce their risk of being infected. U-state node is transformed into an A-state node under one of the following two conditions:

(i)	It is informed by its neighbors in the virtual awareness layer with the probability $\lambda$;

(ii) The node is infected by the virus. 

Notably, the node in the A-state loses awareness due to time and returns to the U-state with the probability $\delta$. The individual may forget or stop caring about it after the corresponding seasonal epidemic occurs.

\begin{figure}[ht]
    \centering
    \includegraphics[scale = 0.38]{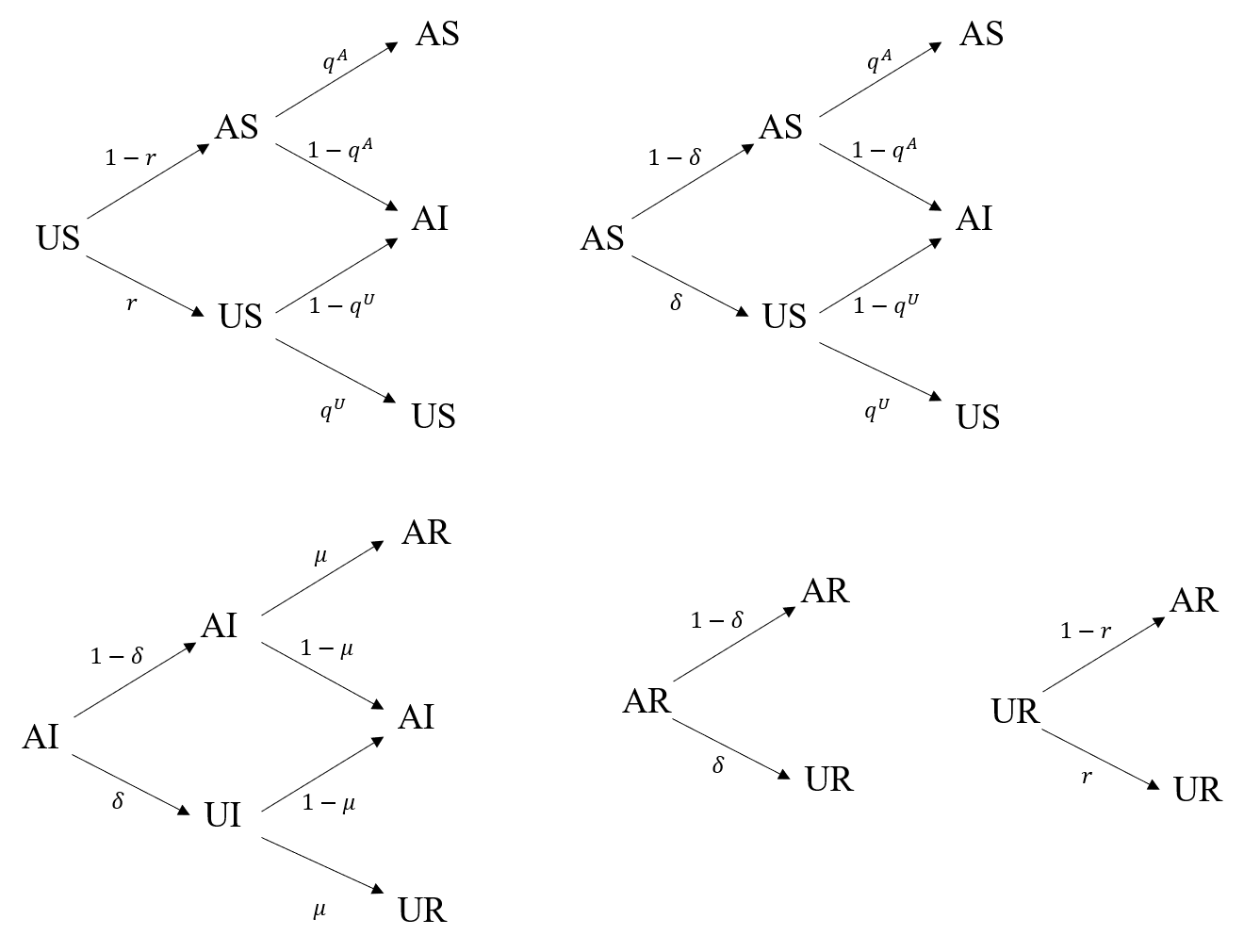}
      
    \caption{\textbf{The transfer probability tree of five states.} It is a visual description of Eq. (1) and Eq. (2), the state that a node is currently in is taken as the root node of this tree, and the probability that it will be in each state at the next time step can be derived from this tree. The first layer of the first tree represents whether an individual in the US-state transitions to A-state, while the second layer represents whether the individual transitions to I-state. The final probability of the transition is the product of the probabilities associated with the two arrows. The remaining trees are analogous to the first one.}
\end{figure}
  
The classical susceptible-infected-recovered (SIR) model is used for epidemic spreading. In the physical contact layer: ``S'' is the susceptible state, ``I'' denotes the infected state, and ``R'' represents the recovered state. If an individual in the susceptible state encounters an individual in the infected state, he/she will be infected with the probability $\beta$. This probability is different from awareness nodes (A-state) and unawareness nodes (U-state), the probability of a susceptible individual in the unawareness state being infected by a node in the infected state is $\beta_{U}$, while the probability of a susceptible individual in the awareness state being infected is $\beta_A = \gamma \cdot \beta_U$, where $\gamma \in [0,1)$. When $\gamma = 0$, it indicates that the awareness node is completely immune to the disease. When a node is infected, it will become aware immediately, i.e., change to the A-state, and will not change until it recovers. At the same time, an individual in the infected state will change to the recovered state with the probability $\mu$. For individuals in the recovered state, he/she will not be infected again and has no ability to infect others.

According to the above description, a node can be in the following five states: unaware and susceptible (US), aware and susceptible (AS), aware and infected (AI), unaware and recovered (UR), and aware and recovered (AR). It is noted that the unaware and infected (UI) state does not exist, because a node will immediately become aware once it is infected. Fig. 2 shows the change pattern between these states more clearly in the form of states transfer probability tree. The denotations of the main parameters are listed in Table 1.

In brief, we construct an awareness-virus co-evolution two-layer network model, in which the awareness layer contains a special kind of nodes ($\Omega$ nodes). They are always in the U-state and do not interact with other nodes.

\section{Theoretical threshold analysis based on MMCA}

Infectious disease threshold analysis is critical to controlling infectious diseases. In this section, we utilize the MMCA to analyze the threshold of infectious disease spreading. The probability tree in Fig. 2 reveals the transfer relationship between different states, and we can get a total of five possible states. We next establish the dynamic equations for the transition between the five states in Lemma 1.

\begin{lemma}
For each node $i$, the dynamic equations of MMCA are

\begin{equation}
    \left\{\begin{aligned}
        P_{i}^{A S}(t+1) &=P_{i}^{A S}(t)(1-\delta) q_{i}^{A}(t)+P_{i}^{U S}(t)[1-r_{i}(t)] q_{i}^{A}(t) \\
        P_{i}^{U S}(t+1) &=P_{i}^{A S}(t) \delta q_{i}^{U}(t)+P_{i}^{U S}(t) r_{i}(t) q_{i}^{U}(t) \\
        P_{i}^{A I}(t+1) &=P_{i}^{A S}(t)\left\{(1-\delta)\left[1-q_{i}^{A}(t)\right]+\delta\left[1-q_{i}^{U}(t)\right]\right\}\\
        &+P_{i}^{U S}(t)\{r_{i}(t)[1-q_{i}^{U}(t)]\\
        &+[1-r_{i}(t)t][1-q_{i}^{A}(t)]\}+P_{i}^{A I}(t)(1-\mu)\\
        P_{i}^{A R}(t+1) &=P_{i}^{A I}(t)(1-\delta) \mu+P_{i}^{A R}(t)(1-\delta)\\
        &+P_{i}^{U R}(t)\left[1-r_{i}(t)\right] \\
        P_{i}^{U R}(t+1) &=P_{i}^{A I}(t) \delta \mu+P_{i}^{A R}(t) \delta+P_{i}^{U R}(t) r_{i}(t)
    \end{aligned}\right.    
\end{equation}
In the awareness layer, the probability that individual $i$ does not switch from the U-state to the A-state at time $t$ is defined as $r_i(t)$; in the disease spreading layer, the probability that individual $i$ is not infected by its neighbors in the $A$ and $U$ states at time $t$ is defined as $q_i^A(t)$, and $q_i^U(t)$.

\end{lemma}

\begin{proof}
Respectively, the probability that $r_i(t)$, $q_i^A(t)$, and $q_i^U(t)$ can be described as follows:
\begin{equation}
    \left\{\begin{array}{l}
        r_{i}(t)=\prod_{j}\left[1-a_{j i} P_{j}^{A}(t) \lambda\right] \\
        q_{i}^{A}(t)=\prod_{j}\left[1-b_{j i} P_{i}^{A I}(t) \beta^{A}\right] \\
        q_{i}^{U}(t)=\prod_{j}\left[1-b_{j i} P_{i}^{A I}(t) \beta^{U}\right]
    \end{array}\right.
\end{equation}
where $a_{ij},b_{ij}$ are the elements in the adjacency matrix $A,B$. If there are connected edges between nodes $i$ and $j$, then $a_{ij}$ or $b_{ij}=1$, otherwise its value is 0.

Primarily, define $A=(a_{ij})_{N\times N}, B=(b_{ij})_{N\times N}$ as the adjacency matrix of the upper (information diffusion) network and the lower (disease propagation) network respectively. In each time step, an individual can only be in one of the five states. Define the proportion of states of individual $i$ at time $t$ as $P_i^{AS}(t), P_i^{AI}(t), P_i^{US}(t), P_i^{UR}(t)$, and $P_i^{AR}(t)$.

Next, to prove Lemma 1, it is essential that the state transfer equation holds, i.e., that all states that can be transferred to the left part of the equation are multiplied by their relative transfer probabilities.

Specifically, as the evolution of all these states of nodes follows a Markov process, the proportion of individuals in the AS-state at the moment $t$ is $P_{i}^{A S}(t)$, the probability that an individual in AS-state has not transformed from the A-state to the U-state is $(1-\delta)$, the probability that individual is not infected is $q_{i}^{A}(t)$.
Thus, the individuals in the AS state at the moment $t$ are still in the AS state at the moment $t+1$ can be calculated by Eq. 3.
\begin{equation}
    P_{i}(t+1) =P_{i}^{A S}(t)(1-\delta) q_{i}^{A}(t)
\end{equation}

Similarly, the proportion of individuals in the US-state at the moment $t$ is $P_{i}^{U S}(t)$, the probability that an individual in US-state is transformed from the U-state to the A-state is $[1-r_{i}(t)]$, the probability that an AS-state individual is not infected is $q_{i}^{A}(t)$.
So the probability of individuals in the US-state at the moment $t$ and in the AS-state at the moment $t+1$ is expressed as
\begin{equation}
    P_{i}(t+1) =P_{i}^{U S}(t)[1-r_{i}(t)] q_{i}^{A}(t)
\end{equation}

Since the probability of transforming to the AS-state is 0 for all other states, combining Eqs. 3 and 4, the individuals in the AS-state at time $t+1$ is
\begin{equation}
    P_{i}^{A S}(t+1) =P_{i}^{A S}(t)(1-\delta) q_{i}^{A}(t)+P_{i}^{U S}(t)[1-r_{i}(t)] q_{i}^{A}(t)
\end{equation}

In conclusion, the first equation in Eq. 1 holds. Analogously, we can prove that the rest equations in Eq. 1 hold.

The result follows.
\end{proof}

Lemma 1 describes the relationship between the probability value of individual $i$ being in each state at the next time step $t+1$ and the probability value of each state in individual $i$ at the current time step $t$. It is the set of propagation dynamics equations for the general case. To find the threshold equation, we assume that propagation occurs near the threshold, then we can say

\begin{theorem}
In our proposed model, as the propagation occurs near the threshold, Lemma 1 can be abbreviated as
\begin{equation}
    \left\{\begin{array}{l}
        P_{i}^{A S} =P_{i}^{U S}\left(1-r_{i}\right)+P_{i}^{A S}(1-\delta) \\
        P_{i}^{U S} =P_{i}^{U S} r_{i}+P_{i}^{A S} \delta \\
        \varepsilon_{i}=\varepsilon_{i}(1-\mu)+P_{i}^{U S}\left[r_{i} \alpha^{U}+\left(1-r_{i}\right) \alpha^{A}\right]\\
        +P_{i}^{A S}\left[\delta \alpha^{U}+(1-\delta) \alpha^{A}\right]
    \end{array}\right.    
\end{equation}
\end{theorem}

\begin{proof}

With the system of equations in Eq. 1, as $t\to\infty$, each state ($P_i^{AS}(t), P_i^{AI}(t), P_i^{US}(t), P_i^{UR}(t), P_i^{AR}(t)$) reaches steady state, we yield

\begin{equation}
    \lim_{t\to\infty}{P_i^*(t+1)}=P_i^*(t)
\end{equation}
where the $*$ sign represents various states, such as $UR, AI, US,$ etc.

Since the individual can only be in one of the five states at each time step, the sum of the proportions of individuals in each state at each time step is 1. When the transmission rate of an infectious disease is near the threshold, the disease does not break out and the probability of an individual being in the infected state tends to be close to 0.
Consequently, we have:
\begin{equation}
    P_i^{AI} = \varepsilon \ll 1
\end{equation}
According to Eq. 8, we obtain:
\begin{equation}
    q_{i}^{A}=\prod_{j}[1-b_{i j} P_{j}^{A I}(t) \beta^{A}] \approx(1-\beta^{A} \sum_{j} b_{i j} \varepsilon)
\end{equation}

Let $\alpha^{A}=\beta^{A} \sum_{j} b_{j i} \varepsilon_{j}, \alpha^{U}=\beta^{U} \sum_{j} b_{j i} \varepsilon_{j}$, since near the threshold $p_i^{AI}\to0$, the proportion of infected state nodes and the recovered state nodes are sufficiently small, i.e., $P_i^{AR}\to0$ and $P_i^{UR}\to0$, then the last two equations of Lemma 1 can be ignored. Combined with Eq. 9, we replace $q_{i}^{A}$ in the third equation in Eq. 1. Thus the third equation of Eq. 6 holds.

Since propagation occurs near the threshold, we have $q_{i}^{A}(t) \to 1$ and $q_{i}^{U}(t)\to 1$, replacing $q_{i}^{A}(t)$, $q_{i}^{U}(t)$ with 1 in Eq.1, the first and second equations of Eq. 6 holds.

The result follows.
\end{proof}

Theorem 1 gives the set of propagation dynamics equations near the threshold, by which we can obtain the threshold equation.
\begin{theorem}
    If the disease recovery rate is $\mu$, and the greatest eigenvalue of the matrix $H (h_{ij}=[1-(1-\gamma)P_i^A]b_{ji})$ is $\Lambda_{\max}$, then the propagation threshold is
    \begin{equation}
    \beta_{c}^{U}=\frac{\mu}{\Lambda_{\max }}
    \end{equation}
\end{theorem}

\begin{proof}
    Simplifying Eq. 6 yields:
\begin{equation}
    \begin{array}{l}
        \mu \varepsilon_{i} =\left[P_{i}^{A} \gamma+\left(1-P_{i}^{A}\right)\right] \beta^{U} \sum_{j} b_{j i} \varepsilon_{j}\\
        =\beta^{U}\left[1-(1-\gamma) P_{i}^{A}\right] \sum_{j} b_{j i} \varepsilon_{j}
    \end{array}   
\end{equation}

Transforming Eq. 11, we obtain Eq. 12:
\begin{equation}
    \sum_{j}\left\{\left[1-(1-\gamma) P_{i}^{A}\right] b_{j i}-\frac{\mu}{\beta^{U}} t_{j i}\right\}\varepsilon_{j}=0
\end{equation}
where $t_{ij}$ is an element of the identity adjacency matrix, and $h_{ij}=[1-(1-\gamma) P_{i}^{A}] b_{j i}$ is an element of the matrix $H$.

By the definition of the matrix eigenvalues, we have 
\begin{equation}
    \{[1-(1-\gamma) P_{i}^{A}] b_{j i} - \Lambda_{\max }t_{ji}\}\varepsilon_{j} = 0
\end{equation}

Combining Eq. 12 we get
\begin{equation}
    \frac{\mu}{\beta^{U}} t_{j i} = \Lambda_{\max }t_{ji}
\end{equation}

Then the threshold problem $\beta = \beta_c^U$ for infectious diseases can be considered as the solution of the $H$ eigenvalue problem, i.e., the smallest $\beta^U$ satisfies Eq. 12.

The result follows.
\end{proof}

Theorem 2 shows that the threshold of the two-layer network propagation model is not only related to the network topology, but also related to the density of individuals in A-state.

\section{Numerical simulation}

From this section, we begin to explore the influence of $\Omega$ nodes by the simulations of proposed models, which are supposed to follow the theory of Sec. \uppercase\expandafter{\romannumeral3}

First of all, considering the model in Sec. \uppercase\expandafter{\romannumeral2}, the two-layer multiplex network is constructed as follows: The information transmission layer is a Barabási–Albert (BA) network, and the disease spreading layer is a Watts–Strogatz (WS) network. The new connection number $m$ of the BA network is 4, the reconnection probability $p$ of the WS network is 0.1, and the average degree $\left\langle k \right\rangle=4$, the total number of nodes $N=10000$, the initial proportion of infected equals to 0.1\%, the disease recovery rate $\mu = 0.06$, the information recovery rate $\delta = 0.04$, the information notification rate $\gamma = 0.04$. Next, we will study the propagation process when $\Omega$ nodes are randomly selected.
\subsection{Impact of randomly select $\Omega$ nodes}

In order to simulate the virus propagation process in a two-layer network, we conduct extensive numerical simulations. Given a disease-awareness transmission network, we first randomly set $\Omega$ nodes in the awareness layer, randomly set 10 initial infected nodes, and the rest of the nodes as susceptible states. Set the initial time step $t=0$, let the network evolve according to the dynamic process of Eq. 2.
After reaching the steady state, we record the results for enough long time, and take the average value after the steady state as the value of the desired statistic to reduce the error. Thereby, we can accurately estimate the epidemic threshold and the final infection density. The results obtained from the numerical simulations are shown in Fig. 3.
\begin{figure}[ht]
    \centering
    \includegraphics[scale =0.75]{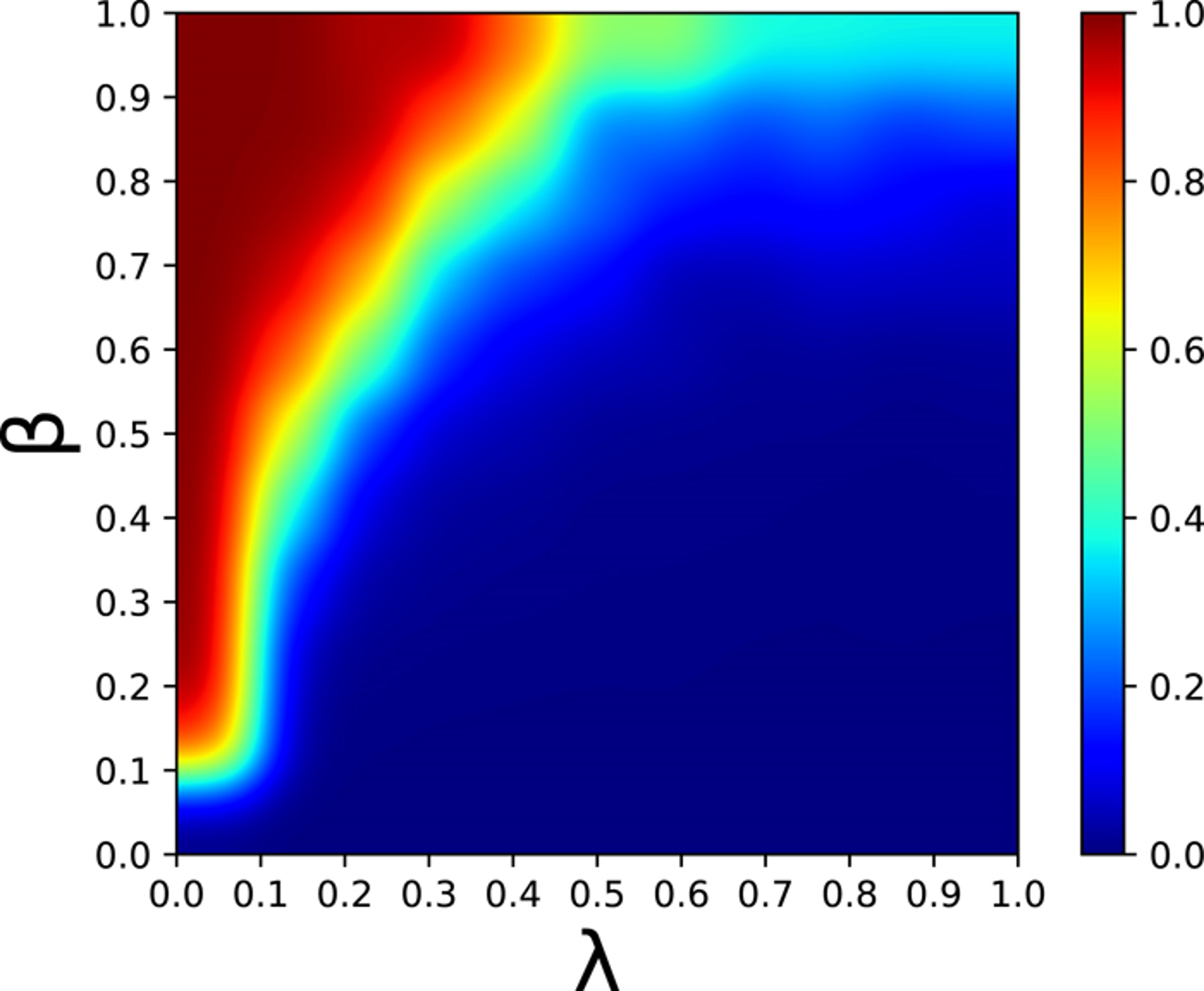}
    
    \caption{\textbf{The proportion of final infections with the information transmission rate and disease spreading rate.} The horizontal axis is the information transmission rate ($\lambda$) and the vertical axis is the disease spreading rate ($\beta$), which refers to the probability of infection of individuals in the unawareness state (U-state), and different colors in the graph mark different proportions of final infections ($\rho^R$), whose specific values are demonstrated in the bar graph on the right. It illustrates that the transmission rate of information has a significant impact on disease propagation and that there exists a clear threshold for disease transmission.}
\end{figure}

As can be seen from Fig. 3, the horizontal coordinates represent the spread rate of information, and the vertical coordinates represent the spread rate of the disease, set them between 0 and 1, and then let the network evolve, with the different colors leaning more towards red representing more infections. By randomly deleting nodes in the process of information dissemination (i.e., permanently setting them to the U-state), we find that in the case where the information dissemination rate $\lambda$ is low, there is a massive outbreak of disease, but as the information dissemination rate $\lambda$ increases, the number of disease infections decreases massively, indicating that the dissemination of information has a great inhibitory effect on the spread of disease. Besides, we can clearly observe a threshold in the propagation of the disease with respect to the parameter $\beta$, which is approximately around 0.1. When the disease transmission rate $\beta$ exceeds the threshold, the disease may erupt in the network; conversely, when the transmission rate $\beta$ is below the threshold, the disease will eventually vanish in the network.

\subsection{The effect of different $\Omega$ nodes on the propagation of two-layer networks}

Subsequently, we will discuss how the selection of different $\Omega$ nodes will affect the network propagation process. We will explore the effect of $\Omega$ nodes on propagation through three different perspectives: degree centrality, betweenness centrality, and clustering coefficient, then compare these properties, and finally lead to a conclusion.

\subsubsection{Impact of degree centrality}

\begin{figure}[ht]
%\begin{tabular}{cc}
\centering
    \includegraphics{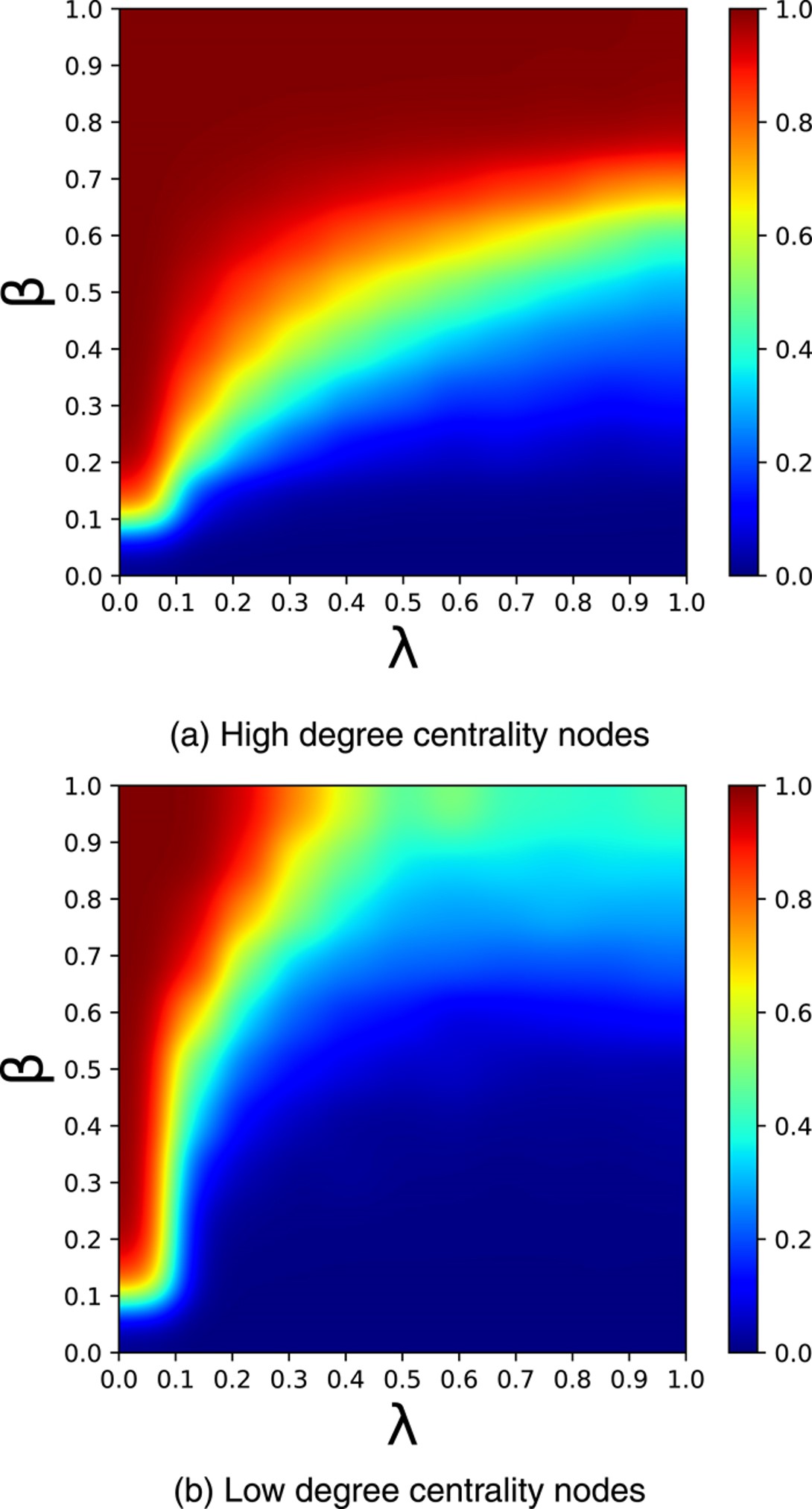}

\caption{\textbf{The proportion of final infected individuals varies with the information transmission rate and disease spreading rate.} the horizontal axis is the information transmission rate ($\lambda$), the vertical axis is the disease spreading rate ($\beta$), which refers to the probability of infection of individuals in the unawareness state (U-state), different colors in the graph mark different proportions of final infected individuals ($\rho^R$), and the experimental results are averaged over 10 simulations. The figure indicates that when $\Omega$ nodes are low-degree centrality nodes, the overall scale of disease propagation is lower than that of high-degree centrality nodes.} 
\end{figure}

Degree is one of the most direct indexes to measure the significance of nodes in epidemic spreading processes. In order to explore the influence of different $\Omega$ nodes on the propagation process in the information network, we first choose the degree of nodes as a measure of $\Omega$ nodes and arrange the nodes according to the magnitude of degree. The 20 nodes with the largest degree are selected, we set them as $\Omega$ nodes in the information dissemination layer, and then take different information dissemination rates and disease dissemination rates. The dissemination results are demonstrated in Fig. 4(a), and then the 20 nodes with the smallest degree are selected and set as $\Omega$ nodes in the information dissemination network. The propagation results are displayed in Fig. 4(b). Combining the two figures, it can be found that the state of the nodes with the large degree in the awareness layer greatly affects the disease spreading process, especially when the disease spreading rate is large, and the nodes with the large degree in the awareness layer will have an antagonistic effect on the disease spreading and inhibit the spreading scale of the infectious disease. When the information transmission rate is low, the transmission effect of both is almost the same.

In order to further explore how the node with a large degree affects the transmission process of infectious diseases, we simulate the transmission process given $\lambda = 0.5, \beta = 0.2,0.5,0.8$, respectively, and let the time be long enough for the number of recovered individuals leveling off, the results are demonstrated in Fig. 5.

\begin{figure}[ht]
\centering
    \includegraphics[scale =0.75]{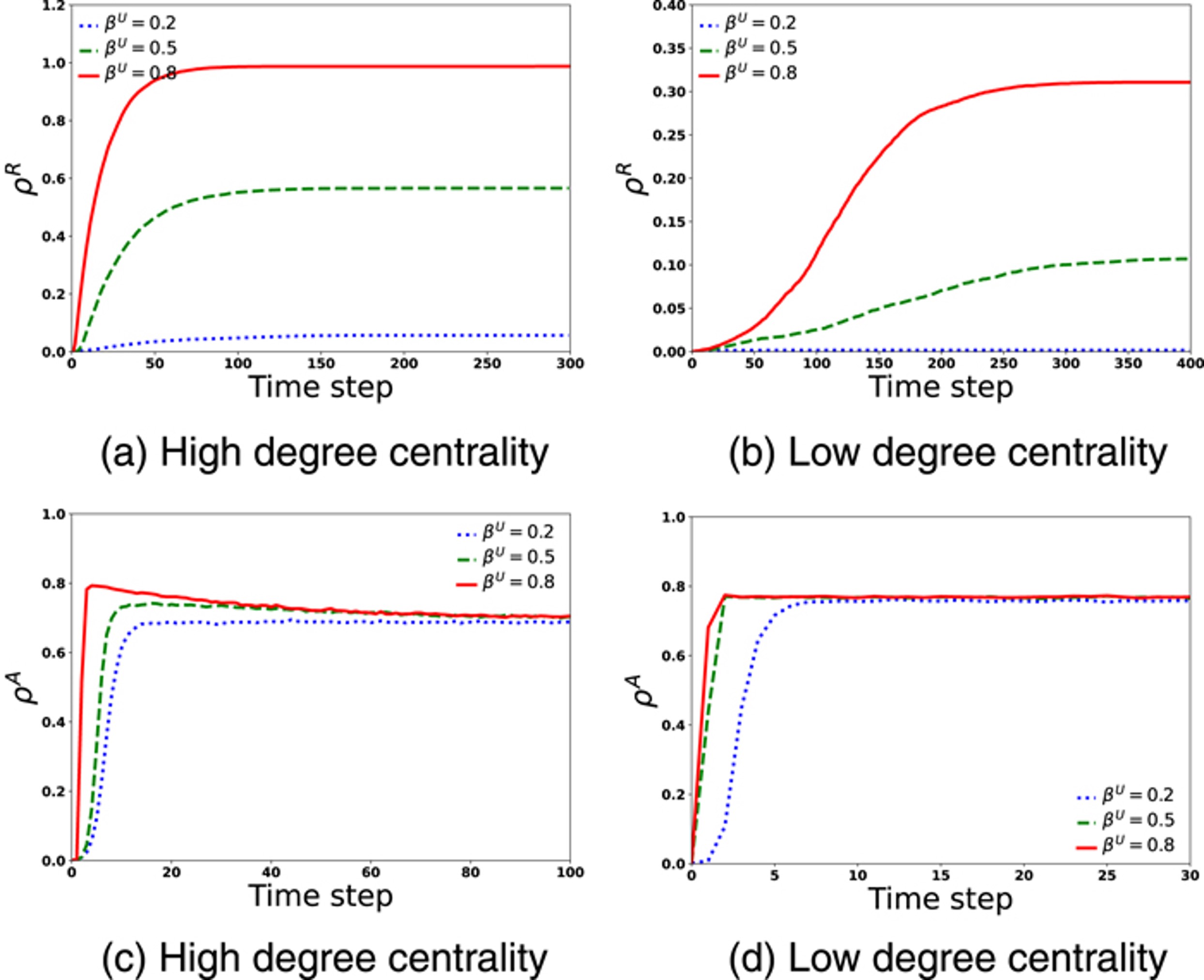}
  
\caption{\textbf{The proportion of $\rho^A$ and $\rho^R$ with different degree centrality $\Omega$ nodes.} (a) and (b) are the curves of the proportion of infected people with time, (c) and (d) are the curves of the corresponding number of people in the awareness state (A-state) with time step. The parameters are set as $N=10000$, $\mu = 0.06$, $\delta = 0.04$, $\gamma = 0.5$, and the experimental results are averaged over 10 simulations. (a) and (b) demonstrate that when $\Omega$ nodes are low-degree centrality nodes, the scale of infection and the time to reach maximum scale are lower than those of high-degree centrality nodes. (c) and (d) indicate that the difference between the two is not significant in terms of information dissemination.} 
\end{figure}

From Figs. 5(a) and 5(b), it can be seen that the disease spreading rate $\beta^U$ greatly affects the transmission process. A larger transmission rate will lead to a larger infection scale ($\rho^R$). Comparing the deletion of nodes with a large degree and a small degree, we can see that after deleting the nodes with a large degree, the scale of disease spreading increases significantly relative to the deletion of nodes with a small degree. Nodes of different degrees will also affect the time when the network reaches a steady state. This phenomenon is reflected in the number of time steps consumed before the number of R-states reaches a steady state, and it can be roughly seen from the Fig. 5 that the time steps required to reach a steady state are around 70 when nodes with a large degree are deleted, while the time steps required to reach steady state are around 250 when nodes with a small degree are deleted. In a side-by-side comparison when $\beta^U = 0.8$, the infection scale is close to 1 for nodes with a large deletion degree, while the infection scale is only about 0.3 for nodes with a small deletion degree, thus it can be concluded that the nodes with the large degree in the information dissemination layer greatly affect the virus propagation process on the two-layer network. From Figs. 5(c) and 5(d), it can be seen that the number of individuals in the A-state will eventually reach a steady state. The number of individuals in the A-state when reaching the steady state is determined by the parameters in the information propagation layer, since the propagation of information can be seen as an independent propagation process. Comparing different $\beta^U$, we can see that the curve with higher $\beta^U$ will reach the peak faster, but the size of the peak will not necessarily be higher than the curve with smaller $\beta^U$. Overall, the number of individuals in the A-state is not greatly affected by $\beta^U$. It is worth noting that in Fig. 5(c), the red curve with $\beta^U$ exhibits a distinct peak, characterized by an initial ascent, followed by a descent, and eventually reaching a stationary state. We hypothesize that the possible reason for this phenomenon is when a node with high centrality acts as the $\Omega$ node, the disease reaches its peak at an exceptionally rapid rate. Consequently, a significant number of individuals get infected in the early stages of disease propagation, and this portion of infected nodes immediately transitions into the A-state. This transition has the potential to surpass the stable state of the information network. As the growth rate of disease transmission decelerates, the density of individuals in the A-state stabilizes around the stable state of the information network.

\subsubsection{Impact of betweenness centrality}
\begin{figure}[ht]
\centering
    \includegraphics{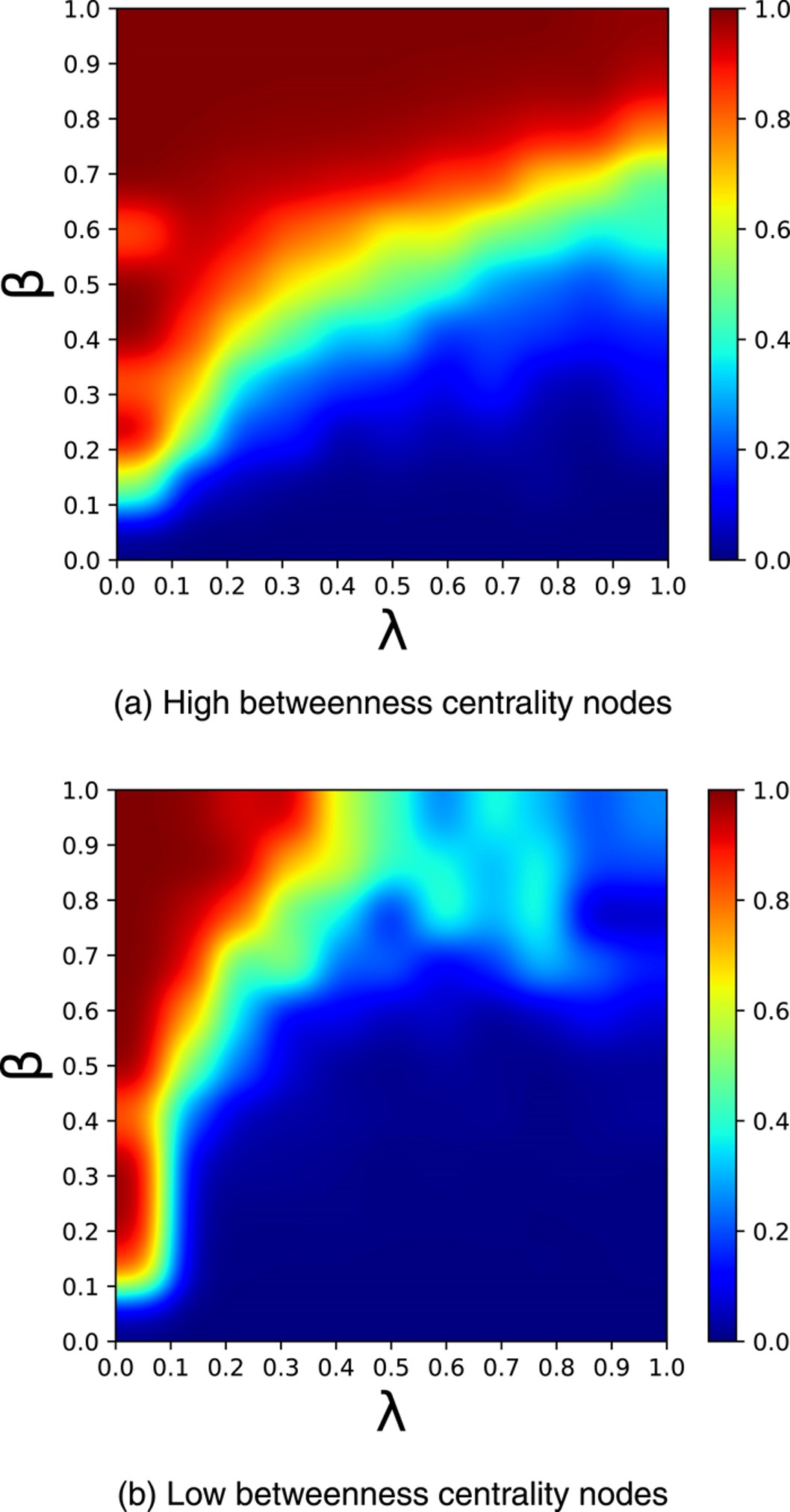}
    
\caption{\textbf{The proportion of final infected individuals varies with the information transmission rate and disease spreading rate.} The horizontal axis is the information transmission rate ($\lambda$), the vertical axis is the disease spreading rate ($\beta$), which refers to the probability of infection of individuals in the unawareness state (U-state), different colors in the graph mark different proportions of final infected individuals ($\rho^R$), and the experimental results are averaged over 10 simulations. This figure demonstrates that nodes with low betweenness centrality have a smaller scale of disease propagation, but both exhibit certain oscillations.} 
\end{figure}

Betweenness centrality is the degree of the number of shortest paths through a node in the network, and it is one of the indicators of node importance portrayal, and the median of node $i$ is defined as
\begin{equation}
    B C_{i}=\sum_{s \neq i \neq t} \frac{n_{s t}^{i}}{g_{s t}}
\end{equation}
where $g_{s t}$ is the number of shortest paths from node $s$ to node $t$, and $n_{s t}^{i}$ is the number of shortest paths through node $i$ among the $g_{s t}$ shortest paths from node $s$ to node $t$. From the above definition, we can see that the higher the betweenness centrality of a node is, the higher the probability of information passing through this node in the information network is, which can be interpreted as a ``traffic hub'' node. So we discuss the impact of nodes with different betweenness centrality on the propagation of the two-layer network in the following, and select the 20\% nodes with the largest betweenness centrality, the result is shown in Fig. 6(a). And then, the 20\% lowest betweenness centrality nodes are set as $\Omega$ nodes. The result is shown in Fig. 6(b).

Combining the two figures, it can be found that the larger betweenness centrality individuals play a great role in the suppression of disease in the awareness layer, especially when the transmission rate of information increases. If the larger betweenness centrality individuals in the information network are always in the U-state, the scale of the disease spreading will increase significantly. For the nodes with the small betweenness centrality removed, when the transmission rate of information and disease spreading rate are both large, compared with Fig. 4(b), the scale of infectious disease spreading shows a certain oscillation. However, all the outcomes still show that the scale of infection increases with the increase of disease spreading rate. The reason for the oscillation is conjectured to be that degree centrality has a more essential influence on the scale of disease infection than betweenness centrality in information networks.

To further explore the propagation process of the betweenness centrality, the propagation process was simulated by taking $\lambda = 0.5, \beta = 0.2,0.5,0.8$, respectively, and let the time be long enough for the number of recovered individuals leveling off, the results are displayed in Fig. 7.

\begin{figure}[t]
\centering
    \includegraphics[scale =0.75]{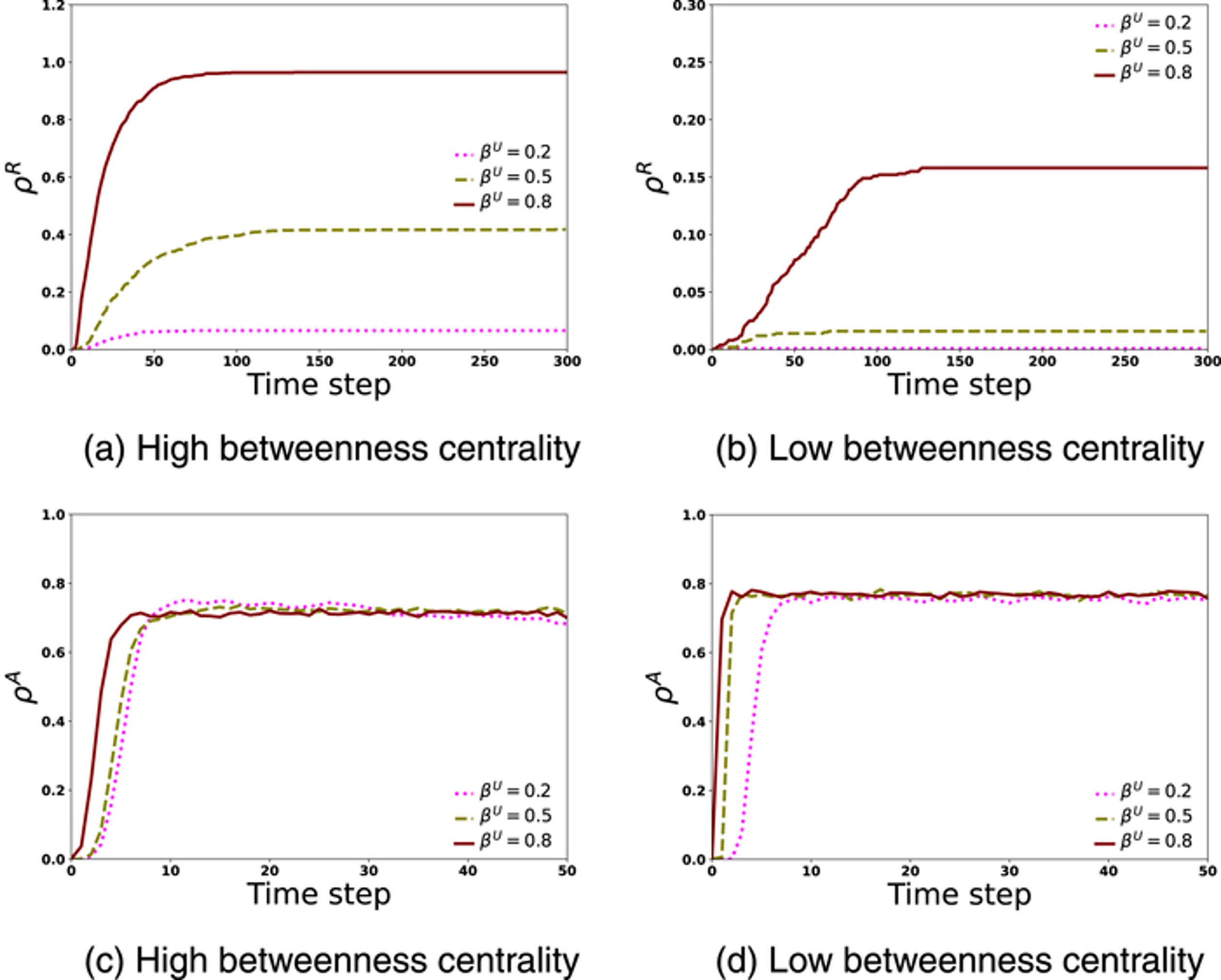}

\caption{\textbf{The proportion of $\rho^A$ and $\rho^R$ with different betweenness centrality $\Omega$ nodes.} (a) and (b) are the curves of the proportion of infected people with time, (c) and (d) are the curves of the corresponding number of people in the awareness state (A-state) with time step. The parameters are set as $N=1000$, $\mu = 0.06$, $\delta = 0.04$, $\gamma = 0.5$, and the experimental results are averaged over 10 simulations. (a) and (b) indicate that the scale and time to reach maximum scale of disease propagation vary with different $\Omega$ nodes. (c) and (d) demonstrate that the larger the value of $\beta^U$, the shorter the time to reach the maximum scale of information dissemination.} 
\end{figure}

From Figs. 7(a) and 7(b), as the information dissemination rate $\lambda=0.5$, the disease dissemination rate $\beta^U$ equally affects the scale of infection of infectious diseases and the time taken to reach a steady state. Comparing the results in Fig. 5, the scale of infection of diseases is less removing small centrality individuals in the intermediary number of information dissemination layer than removing individuals with large degree centrality. From Figs.  7(c) and 7(d), it can be seen that the different values of betweenness centrality do not affect the $\rho^A$ in the awareness layer.

\subsubsection{Impact of clustering coefficient}

In this subsection, we will discuss the effect of the clustering coefficient on disease spreading. In the information network, if a node has a larger clustering coefficient, it is more likely to indirectly notify other individuals about the disease, thus inhibiting the outbreak of the disease. We select the top 20\% nodes with the largest clustering coefficient and the last 20\% nodes with the smallest clustering coefficient, and the results are subsequently presented in Fig. 8.
\begin{figure}[ht]
%\begin{tabular}{cc}
\centering
    \includegraphics{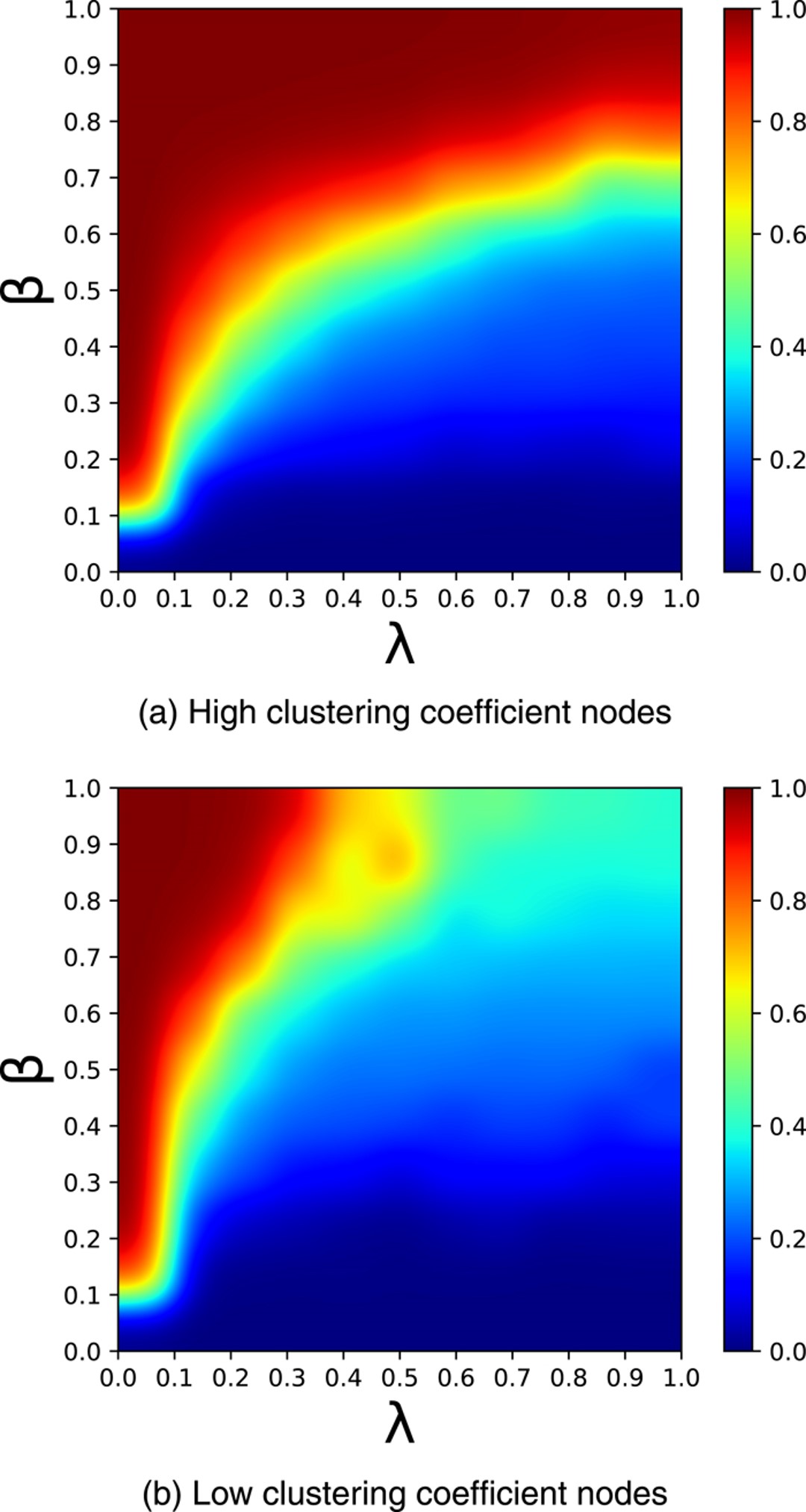}

\caption{\textbf{The proportion of final infected individuals varies with the information transmission rate and disease spreading rate.} The horizontal axis is the information transmission rate ($\lambda$), the vertical axis is the disease spreading rate ($\beta$), which refers to the probability of infection of individuals in the unawareness state (U-state), different colors in the graph mark different proportions of final infected individuals ($\rho^R$), and the experimental results are averaged over 10 simulations. This figure indicates that when $\Omega$ nodes are low clustering coefficient nodes, the overall scale of disease infection is lower. However, compared with Fig. 4 and Fig. 6, the infection scale is relatively larger at low $\beta$.} 
\end{figure}

From Fig. 8, it can be seen that the larger clustering coefficients individuals in the awareness layer also play a great inhibitory role in disease spreading. If they are set as $\Omega$ nodes, they will increase the final infection scale. It is worth mentioning that when the small clustering coefficients nodes are set as $\Omega$ nodes, and the information transmission rate and disease spreading rate are both larger, the final infection scale compared to the mediator centrality and degree centrality as $\Omega$ node is slightly larger. The reason is that the small clustering coefficients individuals in the awareness layer also have a great influence role on the spread of disease. Further, we set $\lambda = 0.5, \beta = 0.2,0.5,0.8$ to study the spread process when nodes with different clustering coefficients are set as $\Omega$ nodes, respectively, and let the time be long enough for the number of recovered individuals leveling off, and the results of the simulation are displayed in Fig. 9.

\begin{figure}[t]
%\begin{tabular}{cc}
\centering
    \includegraphics{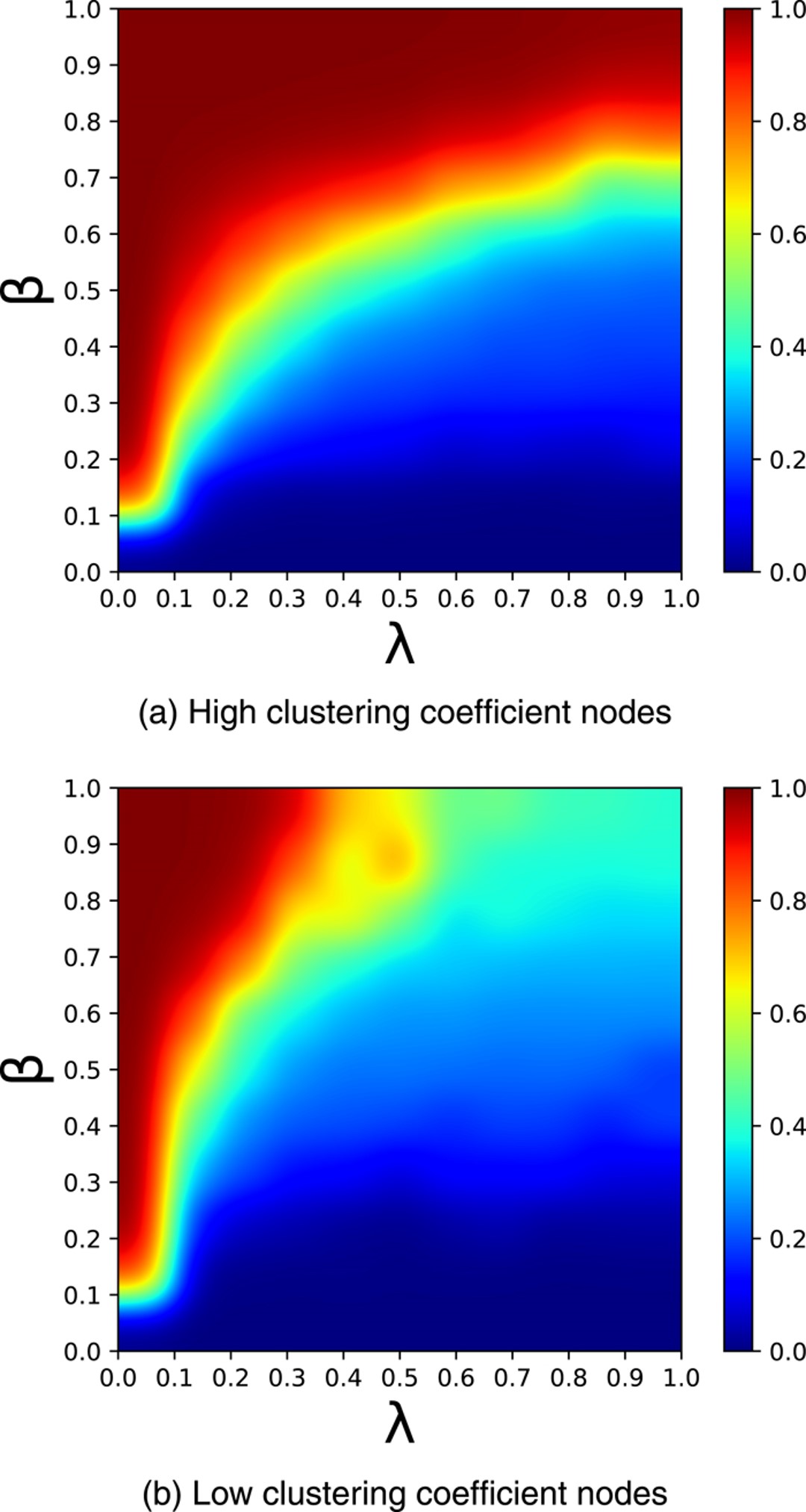}
%\end{tabular}
\caption{\textbf{The proportion of $\rho^A$ and $\rho^R$ with different clustering coefficient $\Omega$ nodes.} (a) and (b) are the curves of the proportion of infected people with time, (c) and (d) are the curves of the corresponding number of people in the awareness state (A-state) with time step. The parameters are set as $N=10000$, $\mu = 0.06$, $\delta = 0.04$, $\gamma = 0.5$, and the experimental results are averaged over 10 simulations. This figure is similar to Fig. 5 and Fig. 7. Overall, low $\beta^U$ will reduce the scale of disease infection, but it has no effect on the scale of information propagation. It only affects the time it takes to reach the maximum scale.} 
\end{figure}

From Figs. 9(a) and 9(b), different clustering coefficient nodes in the awareness layer still have different effects on disease propagation, which are not so significant compared to betweenness centrality and degree centrality. Combined with Figs.  5(c) and 5(d), Figs.  7(c) and 7(d), Figs.  9(c) and 9(d), the number of A-state nodes in the network does not vary by the choice of different $\Omega$ nodes, and different values of $\beta^U$s do not affect the proportion of A-state nodes in the steady state, but only their growth rate.

\subsection{The ratio of $\Omega$ nodes}

\begin{figure*}[ht]
\centering
    \includegraphics[scale =1.5]{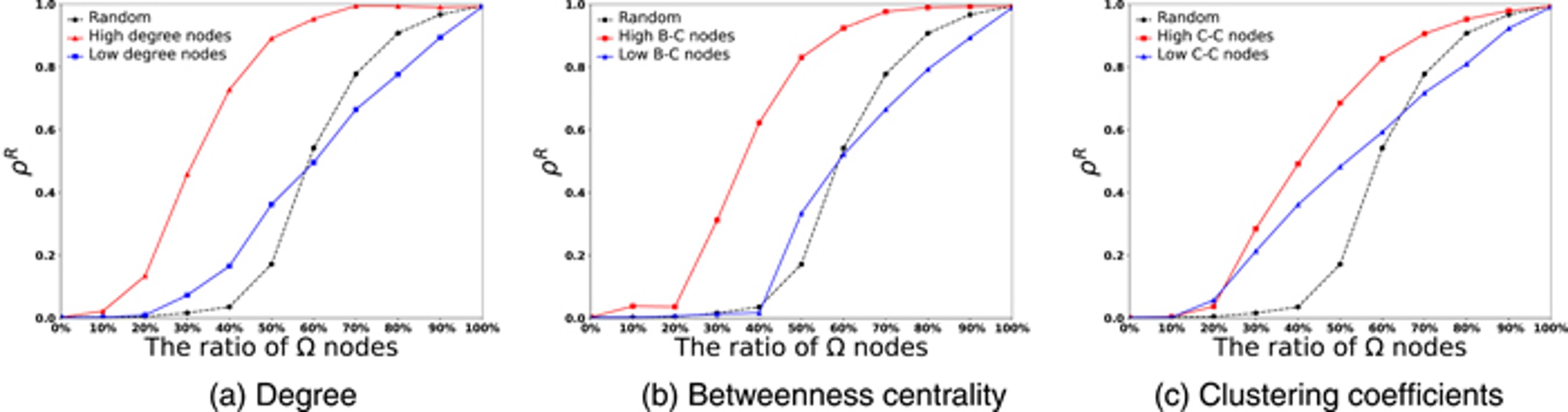}
\caption{\textbf{The final infection density ($\rho^R$) varies with the proportion of $\Omega$ nodes.} Different $\Omega$ nodes are taken to be distinguished by different color curves, the comparison of the same class of attributes is within a subplot, and (a) (b) (c) all contain randomly selected $\Omega$ nodes as a comparison, the experimental parameters are set as follows: $N = 10000$ , $\mu = 0.06$, $\delta = 0.04$, $\gamma = 0.5$, and the experimental results are averaged over 10 simulations. In figure. (b), ``B-C'' represents betweenness centrality, in figure. (c), ``C-C'' represents clustering coefficients. This figure demonstrates that when the $\Omega$ nodes are high centrality nodes, the overall infection size is always higher than others, but when they are low centrality nodes, the infection size shows a linear relationship with the increase of the $\Omega$ node proportion. When the $\Omega$ nodes are randomly selected, the curve shows an S-shape, indicating that low centrality nodes have a lower-order impact on propagation.}
\end{figure*}

Subsequently, we will discuss the effect of different proportions of $\Omega$ nodes on propagation. The initial network is set as $N = 10000$, average degree $\left\langle k \right\rangle = 4$. Information propagation rate is set as $\gamma = 0.3$, infection rate is set as $\beta^U = 0.2, \beta^A = 0.08$. Recovery rate is set as $\mu = 0.06$, and information recovery rate is set as $\delta = 0.04$. We choose different ratios of $\Omega$ nodes to plot the ratio of virus infection density with $\Omega$ nodes as Fig. 10 displayed.

From Fig. 10, it can be seen that in contrast to randomly selected $\Omega$ nodes, the infection density will always be greater when $\Omega$ nodes represent large degree individuals, high betweenness centrality individuals, and high clustering coefficients individuals. When $\Omega$ nodes represent individuals with a small degree, individuals with small betweenness centrality, and individuals with low clustering coefficients, the infection density will present greater if the proportion of $\Omega$ nodes is low, and as the proportion of $\Omega$ nodes increases, the number of infections when $\Omega$ nodes are randomly selected will be greater than the number of infections if they are not randomly selected. Besides, if $\Omega$ nodes are individuals with a small degree, small betweenness centrality, or low clustering coefficients, the growth of the number of infections shows an approximately linear relationship with the growth of the proportion of $\Omega$ nodes. It is speculated that the reason for this phenomenon is that such low centrality nodes have less influence on the spread of the virus on the network, and they are more inclined to be infected as some nodes of high importance are infected when they are infected, and they are usually at the end of the information diffusion chain, i.e., they hardly pass information to the next node, so such nodes show an approximately linear increase in infection density as their proportion increases while they are selected as $\Omega$ nodes.

All three simulations illustrate that nodes with higher centrality in the awareness layer have a more important impact on disease spreading in the physical contact layer. In the simulation of the $\Omega$ nodes ratio, the effect of the low centrality individuals on the spread of disease in the physical contact layer is linear. Moreover, it is not influential as a node with high centrality.

\section{Conclusion and outlook}

In the context of epidemic spreading on an awareness-disease two-layer network, we discuss the effects of the states of individuals with different degree centrality, betweenness centrality, and clustering coefficients in the awareness layer on disease spreading. The analysis of epidemic thresholds is based on MMCA, which quantitatively describes the effects of different factors on disease spreading on a two-layer network. The results of the simulations suggest that individuals with a high degree centrality or high betweenness centrality or large clustering coefficients in the awareness layer have a great influence on the overall disease spreading. If they remain in the unawareness state (U), it decreases the disease spreading threshold, the time it takes for the disease to reach its maximum size, as well as the scale of disease spreading. 

The results of this study enhance the understanding of disease spreading on a two-layer network, especially the role played by different nodes in the awareness layer. Based on our analyses, to curb outbreaks, local authorities should make the existence of an outbreak known at the beginning of the outbreak to all individuals with a large weight in the awareness layer, i.e., some news media departments, some people with large visibility in the network, and broadcasts to people who are not online. These individuals with a large weight in the awareness layer play a significant role in the spread of the epidemic.

However, there are still challenging tasks to address. For propagation on a two-layer network with $\Omega$ nodes, we only experimentally derived the effect of $\Omega$ nodes on disease propagation, which lacks theoretical support. Besides, our model is only discussed under the BA network of the awareness layer and the WS small-world network of the physical contact layer, however, real networks may be more complex. In future work, we will portray the importance of nodes in the awareness layer theoretically. There may be certain nodes in the awareness layer that can broadcast information to all nodes instantaneously, e.g., the government can send information to all people through SMS, how such nodes will affect the spread of disease can be further explored. Moreover, nodes of lower importance in the awareness layer, nodes that are usually at the end of the information dissemination chain, and their impact on the spread of disease are also topics worth discussing.

\section*{Acknowledgment}

Dr. Minyu Feng would like to express his immense gratitude for the supervision under Prof. Jürgen Kurths from 2016 to 2017. He and his research group members delightly dedicate this article to Prof. Kurths, sincerely wish him to have a nice 70th birthday and good health!

We wish to acknowledge the support of Humanities and Social Sciences Fund of Ministry of Education of the People$^{\prime}$s Republic of China with Grant No. 21YJCZH028 and the National Nature Science Foundation of China (NSFC) with Grant No. 62206230.

\end{document}